\documentclass[3p, 11pt]{elsarticle}
\usepackage[utf8]{inputenc}
\usepackage{amssymb}
\usepackage{multirow}
\usepackage{color}
\usepackage{amsmath}
\usepackage{amsthm}
\usepackage{rotating}
\usepackage{graphicx}
\usepackage{nicefrac}
\usepackage{lscape}
\usepackage{booktabs}
\usepackage{multirow}
\usepackage{longtable}
\usepackage[normalem]{ulem}
\useunder{\uline}{\ul}{}
\usepackage{enumerate}
\usepackage{setspace}
\usepackage{float}
\usepackage{adjustbox}
\usepackage[unicode=true, colorlinks=true,linkcolor=black, citecolor=blue, urlcolor=blue]{hyperref}
\usepackage{url}

\newtheorem{proposition}{Proposition}
\renewenvironment{proof}[1][Proof.]
{
    \raggedright
    \textbf{#1}
}
{
    $\square$
}

\title{An efficient branch-and-cut algorithm for the parallel drone scheduling traveling salesman problem}
    \author[add1]{Minh Anh Nguyen}
	\ead{anh.nguyenminh@phenikaa-uni.edu.vn}
	\author[add1]{Hai Long Luong}
	\author[add1]{Minh Ho{\`a}ng H{\`a}\corref{cor1}}
	\ead{hoang.haminh@phenikaa-uni.edu.vn}
	\author[add2]{Ha-Bang Ban}

	\address[add1]{\ ORLab, Faculty of Computer Science, Phenikaa University, Hanoi, Vietnam}
	\address[add2]{\ School of Information and Communication Technology, Hanoi University of Science and Technology, Vietnam}
	\cortext[cor1]{Corresponding author.}

\begin{document}
\begin{frontmatter}

%\textbf{Minh Anh Nguyen, Hai Long Luong, Minh Ho\`ang H\`a*} \\
%ORLab, Faculty of Computer Science, Phenikaa University, Hanoi, VietNam\\
\begin{abstract}
%This paper proposes an efficient branch-and-cut algorithm to exactly solve the parallel drone scheduling traveling salesman problem. The problem is first formulated as a mixed integer linear program with truck-flow variables defined on undirected edges, not on directed arcs as in existing models. The formulation is then strengthened by valid inequalities and the branch-and-cut algorithm is developed. The experimental results show that our algorithm can find optimal solutions for all existing instances, but two in a reasonable running time. To make the problem more challenging for future solution methods, we introduce two new sets of 120 larger instances with the number of customers varying from 318 to 783 and test our algorithm and investigate the performance of state-of-the-art metaheuristics on these instances. We show that the proposed algorithm can steadily solve the instances with up to 400 customers to optimality. Optimal solutions of several cases with 600 and 783 customers are also found by our algorithm. This is the first time problems of such a large size are optimally solved.

We propose an efficient branch-and-cut algorithm to exactly solve the parallel drone scheduling traveling salesman problem. Our algorithm can find optimal solutions for all but two existing instances with up to 229 customers in a reasonable running time. To make the problem more challenging for future methods, we introduce two new sets of 120 larger instances with the number of customers varying from 318 to 783 and test our algorithm and investigate the performance of state-of-the-art metaheuristics on these instances. 

\end{abstract}
\begin{keyword}
Parallel drone scheduling; traveling salesman problem; branch and cut; benchmark instances
\end{keyword}
\end{frontmatter}

\section{Introduction}
The goal of this paper is to study the Parallel Drone Scheduling Traveling Salesman Problem (PDSTSP) defined as follows. We are given a complete undirected graph $G=(V, E)$ with the set of vertices $V=\{0,1,...,n, n+1\}$ and the set of edges $E=\{(i, j): i, j \in V, i < j\}$. Vertex 0 represents the depot while vertex $n+1$ is a copy of the depot vertex and is defined for the modeling purpose. The remaining vertices represent the set of customers $V_c$. At the depot, there is a truck and a set $D$ of $m$ drones available for delivery tasks. Denote a subset of \textit{drone-eligible} customers $V_{d}\subseteq V_c$, whose distance from the depot and parcel weight can be served by drones. The set of customers which must be served by the truck is denoted by $V_{t}=V_c \setminus V_d$. The travel time from node $i$ to node $j$ of the truck is denoted as $t_{ij}$. The time required by a drone to service customer $i$ is denoted as $t'_i$, which is calculated by doubling the time of flying from the depot to customer $i$. The PDSTSP consists in constructing a set of routes for the truck and the fleet of drones such that: (i) every customer vertex is visited exactly once by either truck or drone; (ii) the truck route starts and ends at the depot; (iii) the drone performs back-and-forth trips from the depot to serve several customers (one customer for each single trip); (iv) the completion time required to serve all customers and return all vehicles to the depot, is minimized.

The PDSTSP has received a significant attention from the literature since its first introduction in \citet*{murray2015flying}. The authors propose a Mixed Integer Linear Programming (MILP) formulation to model the problem and design a simple heuristic. The latter method divides the customers into two subsets, one for the truck and one for the drones. Two subproblems: the TSP and the Parallel Machine Scheduling (PMS) problem, defined on two customer sets, are separately solved to construct the truck tour and drone tours, respectively. In addition, an improvement step is performed to reassign customers either to the drone set or to the truck to better balance the truck and drones' activities. The solution approaches are tested on small instances with up to 20 customers.

Many attempts have been done to efficiently solve larger PDSTSP instances. \citet*{mbiadou2018iterative} propose the first metaheuristic, called iterative two-step heuristic method. The first step constructs a customer sequence (also called a giant or TSP tour) that visits all the customers. In the second step, a dynamic programming is used to split the giant tour into a truck tour and a set of trips for the drones, creating a complete PDSTSP solution. The method is tested on newly generated instances with up to 229 customers. \citet*{dell2020matheuristic} later propose a matheuristic in which the authors use the same idea as found in \cite*{mbiadou2018iterative}. However, instead of dynamic programming, a MILP approach is used to handle the second step. The obtained results show that the algorithm is slower but provides better solutions than the heuristic of \cite{mbiadou2018iterative}. %Remarkably, 28 new best known solutions out of the 90 instances have been reported.

A nature-inspired algorithm called Hybrid Ant Colony Optimization (HACO) is proposed in \citet*{Dinh2021}. The algorithm reuses the idea of \cite{mbiadou2018iterative} that represents a PDSTSP solution as a permutation of all customers. However, the authors propose a new dynamic programming with a controllable amount of states while exploring wider solution space to split the giant tour into a complete PDSTSP solution. Several problem-tailored components such as local search and solution construction also play important roles in improving the performance of the algorithm. When being tested on benchmark instances from the literature, the HACO algorithm outperforms the two existing metaheuristics in terms of both running time and solution quality. The most recent metaheuristic for the PDSTSP is a Ruin-and-Recreate algorithm called Slack Induction by String and Sweep Removals (SISSR) introduced in \cite{Nguyen2021}. Exploiting PDSTSP solution characteristics in an effective manner, the metaheuristic manages to introduce ``sufficient" rooms to a solution via new removal operators during the ruin phase. It is expected to enhance the possibilities for improving solutions later in the recreate phase. Also in \cite{Nguyen2021}, the authors introduce a new variant of the PDSTSP called PDSVRP in which they allow multiple trucks and consider the objective of minimizing the total transportation costs. 

\citet*{Raj2021} present a variant of the PDSVRP problem in which they consider the objective function that minimizes the completion time all the vehicles come back to the depot. The authors propose three MILP formulations and develop a branch-and-price (BAP) algorithm and a BAP-based heuristic to solve the problem. The authors report that their exact methods can solve to optimality all the small PDSTSP instances of \cite{murray2015flying} and 3 instances with 52 customers of \cite{mbiadou2018iterative}. Most recently, \citet*{MBIADOUSALEU2021} generalize the PDSTSP problem in the same way as the PDSVRP. The problem is called the Parallel Drone Scheduling Multiple Traveling Salesman Problem. A hybrid Iterated Local Search is proposed. A MILP formulation and a simple branch-and-cut approach are also designed to solve small instances. 

%It is also worth to mention two new metaheuristics for the PDSTSP that have recently introduced, one nature-inspired method called hybrid ant colony optimization \citep{Dinh2021}, and one ruin-and-recreate algorithm \citep{Nguyen2021}.
%The extensive experiments on the benchmark instances show that both algorithms outperform existing approaches in terms of computation time and solution quality to become the state-of-the-art methods.

As mentioned above, most approaches proposed to solve the PDSTSP belong to the class of heuristic methods. To the best of our knowledge, the results of exact algorithms for the problem are quite limited. Only optimal results obtained on small instances with no more than 20 customers and three instances with 52 customers are provided in \cite{murray2015flying, dell2020matheuristic, Raj2021}. Optimal solutions for 87 over 90 benchmark instances with up to 229 customers created by \cite{mbiadou2018iterative} have never been reported.
The main motivation of this paper is to fill this methodology gap by proposing a simple, yet efficient exact branch-and-cut (B\&C) algorithm for the PDSTSP. The experimental results on the currently largest problem instances of \cite{mbiadou2018iterative} show that the developed solution approach can solve all but two instances in a reasonable running time. To provide more challenging instances to test solution methods in the future, we introduce a new set of instances with up to 600 customers. We report the results of the state-of-the-art algorithms obtained on these instances to further investigate their performance.

The remaining of the paper is structured as follows: Section \ref{sec-formulation} describe our MILP formulation along with numerous valid inequalities, upon which the branch-and-cut algorithm is proposed in Section \ref{sec-method}. Section \ref{sec-experiment} reports experimental results and Section \ref{sec-conclusion} concludes our research.

% ---------------------
\section{Mixed integer linear formulation and valid inequalities}
\label{sec-formulation}
The PDSTSP has been formulated as MILP models in \cite{murray2015flying, mbiadou2018iterative, dell2020matheuristic, Raj2021}. A common feature of these models is that they use variables defined on directed arcs to represent the flow of truck tour. In this study, our model is based on the truck-flow variables defined on undirected edges to reduce a half of flow variables. As such, our formulation use four types of variables described as follows:
\begin{itemize}
    \item $y_i$: binary variable equal to 1 if customer $i\in V_c$ is served by the truck, and 0 otherwise;
    \item $x_{ij}$: binary variable equal to $1$ if edge $(i, j)$ is traveled by the truck ($i, j \in V: i < j$).
    \item $z_{ik}$: binary variable equal to $1$ if customer $i \in V_d$ is served by drone $k \in D$.
    \item $t_{max}$: positive real variable representing the completion time at which all vehicles return to the depot.
\end{itemize}

The MILP model for the PDSTSP is formulated as:
{\small
\begin{align}
    \textrm{minimize} \qquad \qquad & t_{max} & \label{obj}\\[-0.2em]
    \textrm{subject to} \qquad
    t_{max} &\geq \sum_{i < j} t_{ij}x_{ij}  \label{c1}\\[-0.5em]
    t_{max} &\geq \sum_{i \in V_d}t'_{i}z_{ik} \quad \forall k \in D \label{c2}\\[-0.6em]
    y_i + \sum_{k \in M} z_{ik} &= 1 \quad  \forall i \in V_d \label{c3}\\[-1em]
    y_i &= 1 \quad  \forall i \in V_t \cup \{0, n+1\}\label{c4}\\[-0.6em]
    \sum_{i > 0} x_{0i}  &= 1  \label{c5}\\[-0.6em]
    \sum_{i < n+1} x_{i(n+1)} &= 1 \label{c6}\\[-0.6em]
    \sum_{i< j} x_{ij} + \sum_{j<k} x_{jk} &= 2y_j \quad \forall j \in V_c \label{c7}\\[-0.6em]
   \sum_{ \substack{i \in V, j \in V \setminus S \\ \text{or } j\in S, i \in V\setminus S}} x_{ij} & \geq 2y_k \quad \forall k \in S, S \subseteq V_c, |S| \geq 3\label{c8}\\[-0.8em]
    x_{ij} &\in \{0, 1\} \quad \forall i, j \in V: i < j \label{c9}\\[-0.4em]
    y_i &\in \{0, 1\} \quad \forall i \in V_c \label{c10}\\[-0.4em]
    z_{ik} &\in \{0, 1\} \quad \forall i \in V_d, k \in D \label{c11}\\[-0.4em]
    t_{max} &\geq 0 \label{c12}
\end{align}
}%
\\[-2em]
The objective function (\ref{obj}) minimizes the time all vehicles coming back to the depot. Constraints (\ref{c1}) and (\ref{c2}) define variable $t_{max}$. More precisely, they provide lower bounds for the objective value in terms of each involved vehicle’s completion time. Constraints (\ref{c3}) require that each drone-eligible customer must be served by either the truck or a drone, while constraints (\ref{c4}) ensure the truck must visit all the customers which cannot be served by the drones. Constraints (\ref{c5}) and (\ref{c6}) make sure that the truck must travel on only one edge linking to the depot and its copy, respectively. Constraints (\ref{c7}) are degree constraints for the truck tour. They imply that if vertex $j$ is visited by the truck, there must be two edges of the truck tour linking to $j$. The connectivity of the truck tour is guaranteed by constraints (\ref{c8}), which are also called the Subtour Elimination Constraints (SECs). They force the presence of at least two edges between any set $S$ and $V \setminus S$, for every subset $S$ of $V_c$ such that $S$ contains at least one vertex visited by the truck. Finally, the remaining constraints define the domain of variables. Note that in this formulation, we allow the case where $x_{0(n+1)}=0$ since this can happen if the drones serve all the customers and the truck does not serve any of them.

The formulation above can be enriched by the presence of valid inequalities. The following constraints follow from the definition of the $x_{ij}$ and $y_k$ variables:\\[-1.8em]
{\small
\begin{align}
    x_{ij} &\leq y_i \quad  \forall i,j \in V: i < j \label{c13}\\
    x_{ij} &\leq y_j \quad  \forall i,j \in V: i < j \label{c14}
\end{align}
}%

The symmetry between the drones can be broken by adding the following constraints: \\[-1.8em]
{\small
\begin{align}
    \sum_{i \in V_d}d_{i}z_{ik-1} &\geq \sum_{i \in V_d}d_{i}z_{ik} \quad  \forall k \in M \setminus \{1\} \label{c15}
\end{align}
}%

Then, the constraints (\ref{c2}) can be replaced by:\\[-1.8em]
{\small
\begin{align}
     t_{max} &\geq \sum_{i \in V_d}d_{i}z_{i1} \label{c16}
\end{align}
}%
\\[-1.2em]
It is possible to derive for the formulation several different forms of the Subtour Elimination Constraints (SECs). Constraints (\ref{c8}) are equivalent to the following constraints:

%3

\begin{proposition}
    The constraints \\[-1.8em]
    {\small
    \begin{align}
        &\sum_{i, j \in S:i<j} x_{ij} \leq  \sum_{k \in S \setminus l} y_k & \forall l \in S, S \subseteq V_c, |S| \geq 3 \label{c17}
    \end{align}
    }%
    are valid SECs for the PDSTSP.
\end{proposition}
\begin{proof}
    Let $S^*$ be the set of vertices visited by the truck in a feasible PDSTSP solution. The left-hand side of (\ref{c17}) is less than or equal to $|S \cap S^*| -1$, while the right-hand side is equal to $|S \cap S*| - y_l$. Because $y_l \leq 1 \; \forall l \in S$, constraints (\ref{c17}) are always valid.
\end{proof}

Summing up constraints (\ref{c17}) for every vertex $l \in S$, we have:

\begin{proposition}
    The constraints \\[-1.8em]
    {\small
    \begin{align}
        &|S| \sum_{i, j \in S:i<j} x_{ij} \leq  (|S|-1) \sum_{k \in S} y_k &\forall S \subseteq V_c, |S| \geq 3 \label{c18}
    \end{align}
    }%
    are valid SECs for the PDSTSP.\\[-1.8em]
\end{proposition}

In a similar manner, the following constraints can be derived directly from constraints (\ref{c8}):

\begin{proposition}
    The constraints \\[-1.8em]
    {\small
    \begin{align}
        &|S| \sum_{\substack{i \in S, j \in V \setminus S \\ \text{or } j\in S, i \in V\setminus S}} x_{ij} \geq 2 \sum_{k \in S} y_k &\forall S \subseteq V_c, |S| \geq 3 \label{c19}
    \end{align}
    }%
    are valid SECs for the PDSTSP. \\[-1.8em]
\end{proposition}

Connectivity constraints (\ref{c8}) and (\ref{c17}) can be strengthened if there exists at least one customer that must be serviced by the truck.  

\begin{proposition}
    The constraints \\[-1.8em]
    {\small
    \begin{align}
        & \sum_{\substack{i \in S, j \in V \setminus S \\ \text{or } j\in S, i \in V\setminus S}} x_{ij} \geq 2 &\forall S \subseteq V_c, S \cap V_t \neq \varnothing, |S| \geq 3 \label{c20}
    \end{align}
    }%
    are valid SECs for the PDSTSP.
\end{proposition}

\begin{proof}
    Immediate since in this case the truck tour in the solution must contain at least one vertex in $S$ and one vertex in $V\setminus S$.
\end{proof}

The equivalent form of constraints (\ref{c20}) can be written as follows:

\begin{proposition}
    The constraints \\[-1.8em]
    {\small
    \begin{align}
        &\sum_{i, j \in S:i<j} x_{ij} \leq \sum_{k \in S} y_k - 1 &\forall S \subseteq V_c, S \cap V_t \neq \varnothing, |S| \geq 3 \label{c21} 
    \end{align}
    }%
    are valid SECs for the PDSTSP.
\end{proposition}

\begin{proof}
    Immediate from constraints (\ref{c17}) since in this case there must be a vertex $l \in S$ such that $y_l = 1$.
\end{proof}

Several other valid inequalities for the symmetric TSP can be applied directly to the PDSTSP (see \cite{Padberg}) or can be strengthened. In this study, we use valid inequalities based on the generation of the 2-matching constraints. 

% We combine $v_0$ and $v_{n+1}$ into $D$ ($D$ is considered as one element).
% Every variables $x$ of $D$ are similar to sum of corresponding variables $x$ of $v_0$ and $v_{n+1}$.
% Variable $y$ of $D$ is similar to variable $y$ of $v_0$.
% We have $N' = V \cup D$.

\begin{proposition}
    The constraints \\[-1.8em]
    {\small
    \begin{align}
        &\sum_{i, j \in H:i<j} x_{ij} + \sum_{(i, j) \in E'} x_{ij} \leq \sum_{j \in H} y_j + \frac{1}{2}(|H'| - 1) \label{2mc}
    \end{align}
    }%
    for all $H \subseteq V_c$, $H' \subseteq H$ and $E' \subset E$ satisfying: (1) $|H'| \geq 3$ and odd; (2) only one endpoint of each edge in $E'$ is in $H'$ and each vertex in $H'$ must be an endpoint of a single edge in $E'$; (3) there is no edge in $E'$ such that whose both endpoints are in $H$; and (4) two different edges in $E'$ do not share a same endpoint.
    
    are valid inequalities for the PDSTSP.
\end{proposition}

\begin{proof}
    Summing up constraints (\ref{c7}) for all $j \in H$, we obtain:
    {\small
    \begin{align}
        &2 \sum_{i, j \in H:i<j} x_{ij} + \sum_{\substack{i \in H, j \in V \setminus H \\ \text{or } j\in H, i \in V\setminus H}} x_{ij} = 2 \sum_{j \in H} y_j \label{proof3}
    \end{align}
    }%

    We add $\sum_{(i, j) \in E'} x_{ij}$ to the left-hand side of (\ref{proof3}) and $|H'|$ to the right-hand side.
    Since $x_{ij} \leq 1$ due to (\ref{c9}) and two different edges in $E'$ do not share a same endpoint, equation (\ref{proof3}) becomes the following inequality: \\[-1.8em]
    {\small
    \begin{align}
        &2 \sum_{i, j \in H:i<j} x_{ij} + 2 \sum_{(i, j) \in E'} x_{ij} + \sum_{\substack{i \in H \setminus H', j \in V \setminus H\\ \text{or } j \in H \setminus H', i \in V \setminus H }} x_{ij} + \sum_{\substack{i \in H', j \in V \setminus H: (i, j) \notin E'\\ \text{or } j \in H', i \in V \setminus H: (j, i) \notin E'}} x_{ij} \leq 2 \sum_{j \in H} y_j + |H'|\label{proof4} 
    \end{align}
    }%

    Removing the third and fourth terms of the left-hand side from (\ref{proof4}) and dividing by $2$, we obtain: \\[-1.7em]
    {\small
    \begin{align}
        & \sum_{i, j \in H:i<j} x_{ij} + \sum_{(i, j) \in E'} x_{ij}  \leq \sum_{j \in H} y_j + \frac{1}{2}|H'|\label{proof5}
    \end{align}
    }%

    As the left-hand side of (\ref{proof5}) is integer and $|H'| \geq 3$ and $|H'|$ is odd, we can round down the right-hand side, which yields (\ref{2mc}).
\end{proof}

\section{Branch-and-cut algorithm}
\label{sec-method}
The PDSTSP is solved to optimality using a standard branch-and-cut algorithm. We solve a linear programming (LP) including the objective function (\ref{obj}) and the constraints (\ref{c1}), (\ref{c13}), (\ref{c14}), (\ref{c15}), (\ref{c16}), and (\ref{c3})-(\ref{c7}). We then search for violated constraints of type (\ref{c8}) and (\ref{c17})-(\ref{2mc}), and the detected constraints are added to the current LP, which is then reoptimized. This process is repeated until all the constraints are satisfied. If there are fractional variables, we branch. If all the variables are integer, we explore another node.

The separation of the connectivity constraints (\ref{c8}) and (\ref{c17})-(\ref{c21}) is performed as follows. From a LP solution $(x^*, y^*)$ obtained on a node of the search tree, we build a graph $G^* = (V^*, E^*)$ where set $V^*$ includes vertex $i$ if the value of variable $y^*_i$ is positive and set $E^*$ contains edge ($i, j$) with the cost equal to $x^*_{ij}$. To find set $S \subset V_c$ that violates the connectivity constraints, we first use the Breadth First Search (BFS) to find all the strongly connected components of the graph induced by the current LP solution. If the number of components is at least two, each component is recorded as a set $S$ that we need to find. On the other hand, if only one component exists and the LP solution is still fractional, we solve the minimum cut problem using the procedure described in \cite{Padberg}. The obtained cut separates the vertices of the induced graph into two subsets, we record one which does not contain the depot and violates one of SECs.

For each recorded set $S$, if it includes any customers in set $V_t$, we generate SECs (\ref{c20}) and (\ref{c21}) from $S$. Otherwise, SECs (\ref{c8}), (\ref{c17}), (\ref{c18}), (\ref{c19}) are generated. To limit the number of generated constraints of type (\ref{c8}) and (\ref{c17}), we only apply them on integer solutions. Finally, we notice that constraints (\ref{c8}) and (\ref{c17}), (\ref{c18}) and (\ref{c19}), (\ref{c20}) and (\ref{c21}) are pairwise equivalent. Therefore, it is sufficient to generate only one of them whenever needed. We select the constraint with the smaller number of nonnegative coefficients on the left-hand side in the hope of creating the lighter model that is easier to solve. This idea is also used to propose the Integer Linear Program based method tackling the TSP in \cite{Oswin}. 

%With graph that we consume from the solution, we search for cuts having total capacity less than $2$, which has possibility leading us to violated SECs.
% We have to decide which SECs will be added in order to get rid of unnecessary SECs.

% We have expressions to calculate number of variables in each SECs:
% \begin{itemize}
%     \item SECs (19): $\frac{|S| \times (|S| - 1)}{2} + |S|$
%     \item SECs (20): $|S| \times (n - |S|) + |S|$
%     \item SECs (18): $\frac{|S| \times (|S| - 1)}{2} + |S| - 1$
%     \item SECs (9): $|S| \times (n - |S|) + 1$
%     \item SECs (22): $\frac{|S| \times (|S| - 1)}{2} + |S|$
%     \item SECs (21): $|S| \times (n - |S|)$
% \end{itemize}

% \subsection{Separation of the 2-matching inequalities}
When the BFS mentioned above cannot find any connected components and the solution is still fractional, a search for violated 2-matching constraints (\ref{2mc}) is started. One can use the procedure proposed in \cite{Padberg} running in $\mathcal{O}(n^4)$ to separate these constraints, which is quite expensive. In this study, we propose a faster procedure described as follows.

% \begin{enumerate} [Step 1:]
%     \item Let $(x^*, y^*, z^*, t^*)$ be a fractional solution to be separated. Let $G^* = (V^*,E^*)$ be the weighted graph induced by this variable such that  $G^*$ in which $i \in $. Set every vertex $i \in N'$ unmarked, set $p_i$ as an invalid value. Create three empty sets $E'$, $H$ and $H'$. \label{s1}
%     \item Access every edge $(i, j) \in E$ in decreasing order, if both $p_i$ and $p_j$ are invalid values, set $p_i = j$, $p_j = i$, put $(i, j)$ into $E'$. \label{s2}
%     \item Chose a random vertex $i \in N'$ which has not been marked, mark $i$.
%     If there is no satisfy vertex, the search has been done. \label{s3}
%     \item Put $i$ into set $H$. Using BFS to find out every vertexes $j$ which have not been marked and connect to $i$ through a path not contain any edges of $E'$, mark them, put them into set $H$. \label{s4}
%     \item For every edge of $E'$, which has only one vertex in $H$, put that vertex into $H'$.
%     If the size of $H'$ is odd, add TMCs (\ref{2mc}) with pair of set $H$ and $H'$.
%     Clear set $H$ and $H'$, go to Step \ref{s3}. \label{s5}
% \end{enumerate}

\begin{enumerate} 
    \item Let $G^* = (V^*,E^*)$ be the weighted graph induced by the fractional solution that needs to be separated. First, four empty sets $\overline{E}$, $E^\prime$, $H$, and $H^\prime$ are created. Every vertex $i \in V^*$ is set unmarked.    \label{s1}\\[-1.8em]
    \item We iterate over every edge $(i, j) \in E^*$ in the decreasing order of value $x^*_{ij}$. If both vertices $i$ and $j$ are not the endpoints of an edge in $\overline{E}$, we put edge $(i, j)$ into $\overline{E}$. The aim of this step is to greedily find set $E'$ to possibly maximize the left-hand side of (\ref{2mc}). \label{s2}\\[-1.8em]
    \item We search for a unmarked vertex $i \in V^*$. If we cannot find any unmarked vertex, the process is terminated. Otherwise, we mark $i$ and go to Step \ref{s4}. \label{s3}\\[-1.8em]
    \item Vertex $i$ is put into set $H$. A BFS is triggered to find every unmarked vertex connecting to $i$ through a path which does not contain any edges of $\overline{E}$. All the vertices on the path are marked and put into set $H$. \label{s4}\\[-1.8em]
    \item Set $H'$ is defined as $H' = H \cap V(\overline{E}) $, where $V(\overline{E})$ is the set of vertices which are the endpoints of all edges in $\overline{E}$. Then, set $E'$ is constructed such that $E' \subset \overline{E}$, $|E'| = |H'|$, and only one endpoint of every edge in $E'$ is a vertex in $H'$. If the size of $H'$ is odd and larger than 2, we check the violation of the constraint (\ref{2mc}) associated with sets $H$, $H'$, and $E'$. We then clear sets $H$, $H'$ and $E'$, and go to Step \ref{s3}. \label{s5}
\end{enumerate}

It is trivial that the complexity of steps \ref{s1}, \ref{s3}, and \ref{s5} is $\mathcal{O}(n^2)$. In step \ref{s2}, ordering the edges in $E^*$ takes $\mathcal{O}(n^2\log(n^2))$ and we must iterate over all edges, the overall complexity is thus $\mathcal{O}(n^2\log(n))$. The BFS procedure in step \ref{s4} runs in $\mathcal{O}(n^2)$ because a vertex in $V^*$ is never considered more than once. Hence, the complexity of the separation procedure is $O(n^2\log{n})$, which is lower than $O(n^4)$ of the procedure proposed in \cite{Padberg}.

\section{Computational experiments}
\label{sec-experiment}
%\subsection{Instances and experimental settings}
We test the B\&C algorithm on the existing benchmark instances introduced in \cite{murray2015flying} and \cite{mbiadou2018iterative} with up to 229 customers.
Two harder instance sets are also created to further investigate the performance of our B\&C algorithm as well as the state-of-the-art metaheuristics. For the first set, we follow the same procedure of \cite{mbiadou2018iterative}, but based on larger TSP instances from the TSPLIB library with up to 783 customers: lin318, pr439, rat575, and rat783. From each selected TSP instance, we derive 15 PDSTSP instances by modifying the following parameters: the position of the depot ($dp$), the percentage of drone-eligible customers ($el$), the drone speed ($sp$), and the number of drones (\#). The second set is generated by the instance generator proposed in \cite{Nguyen2021}. Each instance is characterized by: the number of customers ($n$), which is either 400 or 600; the distribution of customer positions ($cd$), which is either random (r), clustered (c), random clustered (rc); the position of the depot ($dp$), which is either in the center of the customers (c) or in one corner of the customers' region (e); the number of drones ($m$), that is between 1 and 5. For each value of $n$, 30 instances of this class are generated.

Our branch-and-cut algorithm is implemented in C++ compiled with GCC 9.3.0 and built around CPLEX 12.10. All the CPLEX parameters are set to their default values, except the relative MIP gap tolerance \texttt{MIPGap} set to $10^{-6}$ to reach the precision of results reported by the existing metaheuristics. All tests are performed on a desktop operating Kubuntu 20.04, and equipped with an AMD Ryzen 3700X processor and 32GB RAM. We employ single thread for reproducibility purposes. The CPLEX is configured to terminate after three hours of solving an instance. We compare the results of the B\&C with two published state-of-the-art methods, HACO in \citep{Dinh2021} and SISSR in \citep{Nguyen2021}. Each metaheuristic is run 10 times on each instance and the best solution is recorded. We set each run of the metaheuristics to stop after 300 seconds on instances with less than 300 customers, and 600 seconds on the larger instances. The best solutions found by the heuristics are also used to provide good upper bounds for the exact method. The detailed benchmark instances and results of our experiments can be found at \texttt{\href{http://orlab.com.vn/home/download}{http://orlab.com.vn/home/download}}.

\begin{table}
\caption{Experimental results of the B\&C algorithm}
\label{tab:res}
\centering
\scriptsize
\begin{tabular}{lllc|cccc}
\hline 
\multicolumn{4}{c|}{Instance} &  \# Opt  & \# Imp  & gap & time\tabularnewline
Name & $n$ & \# Ins & Source &  &  & (\%) & (s)\tabularnewline
\hline 
 & 10 & 360 & \cite{murray2015flying}& 360 & 0 & 0.00 & 0.01\tabularnewline
 & 20 & 360 & \cite{murray2015flying} &  360 & 0 & 0.00 & 0.06\tabularnewline
\hline 
att48 & 48 & 15 & \cite{mbiadou2018iterative} & 15 & 0 & 0.00 & 0.89\tabularnewline
berlin52 & 52 & 15 & \cite{mbiadou2018iterative} & 15 & 0 & 0.00 & 0.71\tabularnewline
eil101 & 101 & 15 & \cite{mbiadou2018iterative} & 15 & 0 & 0.00 & 12.01\tabularnewline
gr120 & 120 & 15 & \cite{mbiadou2018iterative}& 14 & 0 & 0.00 & 740.17\tabularnewline
pr152 & 152 & 15 & \cite{mbiadou2018iterative} & 15 & 0 & 0.00 & 72.95\tabularnewline
gr229 & 229 & 15 & \cite{mbiadou2018iterative}& 14 & 0 & 0.00 & 1680.39\tabularnewline
\hline 
lin318 & 318 & 15 & this paper & 10 & 4 & 0.01 & 4998.34\tabularnewline
pr439 & 439 & 15 & this paper & 14 & 7 & 0.00 & 3241.04\tabularnewline
rat575 & 575 & 15 & this paper & 7 & 12 & 0.20 & 7885.87\tabularnewline
rat783 & 783 & 15 & this paper & 1 & 5 & 0.42 & 10704.08\tabularnewline
\hline 
 & 400 & 30 & this paper & 16 & 20 & 0.05 & 5940.12\tabularnewline
 & 600 & 30 & this paper & 3 & 16 & 0.23 & 10056.52\tabularnewline
\hline 
\end{tabular}
\end{table}

The obtained results are summarized in Table \ref{tab:res} where Column `Name' shows the name of the instance class (if any), Column `$n$' reports the number of customers, Column `\# Ins' represents the number of instances in the class, and Column `Source' shows the reference where the instances are generated. The performance of the algorithm can be observed via the number of successfully solved instances (Column `\# Opt'), the number of instances in which the B\&C algorithm improves the results of the state-of-the-art metaheuristics (Column `\# Imp'), the averaged gap values returned by CPLEX (Column `gap'), and the running time on average of the algorithm in seconds (Column `time'). 

As can be seen in the table result, the B\&C algorithm can solve to optimality all the small instances of \cite{murray2015flying} in very short running times (less than 1 second in all cases). It is worth recalling that, the best branch-and-cut algorithm in the literature, which is proposed in \cite{Raj2021}, can also successfully solve all of these instances. However, they took 0.86 seconds on average to find the optimal solution on the 20-customers instances compared to our method which took 0.06 seconds on average. To solve the hardest instance, they took 47.95 seconds while we took 0.68 seconds, which is much faster. Although their algorithm is run on an Intel i7-6700 processor, which has 1.2 times lower single-thread performance compared to our computer (Passmark Software 2021), our algorithm is clearly faster.

For the largest existing instances introduced in \cite{mbiadou2018iterative}, our algorithm provides optimal solutions for 88 over 90 instances. For the two instances that we cannot solve to optimality, the gaps are quite small (less than 0.01\%). In \cite{Raj2021}, the authors also test their exact methods on a subset of 18 instances selected from 90 instances of \cite{mbiadou2018iterative}, and only three optimal solutions for the berlin52-based instances are found. This again confirms the good performance of our B\&C.

For the first class of the new instances, our algorithm is capable to optimally solve all pr439-based instances but one, two third lin318-based instances, and nearly a half of rat575-based instances. On the largest instances with 783 customers, our algorithm can successfully solve only one instance. The second class of the new instances tends to be more difficult for the algorithm since we can solve to optimality only 16 over 30 400-customer instances and 3 over 30 600-customer instances. All the valid cuts of the algorithm are activated during the search and constraints (\ref{c8}), (\ref{c20}), (\ref{c21}), and (\ref{2mc}) are created more extensively. Regarding the performance of the two metaheuristics, we observe that SISSR dominates HACO in terms of solution quality. The performance of both heuristic methods tends to decrease markedly when the instance size increases since the exact method can improve 64 solutions of the metaheuristics, all improved cases appear on the new instances. Therefore, there is a huge room for designing new metaheuristics in future researches to efficiently handle large PDSTSP instances.  

\section{Conclusions}
\label{sec-conclusion}

In this research, we propose a simple, yet efficient branch-and-cut algorithm to exactly solve the PDSTSP. The problem is first formulated as a mixed integer linear program with truck-flow variables defined on undirected edges, not on directed arcs as in existing models. The formulation is then strengthened by valid inequalities and the branch-and-cut algorithm is developed. The results obtained on the existing benchmark instances show a good performance of the algorithm when it can successfully solve 88 over 90 instances. We also introduce two new sets of larger instances with up to 783 customers. When being tested on these instances, our exact algorithm can successfully solve 51 over 120 instances. The results also allow to investigate the performance of the state-of-the-art metaheuristics. The SISSIR algorithm is the best metaheuristic and performs well on the instances of the literature since it can find all optimal solutions provided by the B\&C. However, to efficiently solve larger instances including ones introduced in this paper, it is necessary to design a new metaheuristic in future researches.

% \section*{References}
\bibliography{ref}

\begin{thebibliography}{9}
\expandafter\ifx\csname natexlab\endcsname\relax\def\natexlab#1{#1}\fi
\providecommand{\url}[1]{\texttt{#1}}
\providecommand{\href}[2]{#2}
\providecommand{\path}[1]{#1}
\providecommand{\DOIprefix}{doi:}
\providecommand{\ArXivprefix}{arXiv:}
\providecommand{\URLprefix}{URL: }
\providecommand{\Pubmedprefix}{pmid:}
\providecommand{\doi}[1]{\href{http://dx.doi.org/#1}{\path{#1}}}
\providecommand{\Pubmed}[1]{\href{pmid:#1}{\path{#1}}}
\providecommand{\bibinfo}[2]{#2}
\ifx\xfnm\relax \def\xfnm[#1]{\unskip,\space#1}\fi
%Type = Article
\bibitem[{Dell'Amico et~al.(2020)Dell'Amico, Montemanni \&
  Novellani}]{dell2020matheuristic}
\bibinfo{author}{Dell'Amico, M.}, \bibinfo{author}{Montemanni, R.}, \&
  \bibinfo{author}{Novellani, S.} (\bibinfo{year}{2020}).
\newblock \bibinfo{title}{Matheuristic algorithms for the parallel drone
  scheduling traveling salesman problem}.
\newblock {\it \bibinfo{journal}{Annals of Operations Research}\/},  {\it
  \bibinfo{volume}{289}\/}, \bibinfo{pages}{211--226}.
%Type = Inbook
\bibitem[{Dinh et~al.(2021)Dinh, Do \& H\`{a}}]{Dinh2021}
\bibinfo{author}{Dinh, Q.~T.}, \bibinfo{author}{Do, D.~D.}, \&
  \bibinfo{author}{H\`{a}, M.~H.} (\bibinfo{year}{2021}).
\newblock \bibinfo{title}{Ants can solve the parallel drone scheduling
  traveling salesman problem}.
\newblock In {\it \bibinfo{booktitle}{Proceedings of the Genetic and
  Evolutionary Computation Conference}\/} (p. \bibinfo{pages}{14–21}).
\newblock \bibinfo{address}{New York, NY, USA}: \bibinfo{publisher}{Association
  for Computing Machinery}.
%Type = Article
\bibitem[{Mbiadou~Saleu et~al.(2018)Mbiadou~Saleu, Deroussi, Feillet, Grangeon
  \& Quilliot}]{mbiadou2018iterative}
\bibinfo{author}{Mbiadou~Saleu, R.~G.}, \bibinfo{author}{Deroussi, L.},
  \bibinfo{author}{Feillet, D.}, \bibinfo{author}{Grangeon, N.}, \&
  \bibinfo{author}{Quilliot, A.} (\bibinfo{year}{2018}).
\newblock \bibinfo{title}{An iterative two-step heuristic for the parallel
  drone scheduling traveling salesman problem}.
\newblock {\it \bibinfo{journal}{Networks}\/},  {\it \bibinfo{volume}{72}\/},
  \bibinfo{pages}{459--474}.
%Type = Article
\bibitem[{{Mbiadou Saleu} et~al.(2021){Mbiadou Saleu}, Deroussi, Feillet,
  Grangeon \& Quilliot}]{MBIADOUSALEU2021}
\bibinfo{author}{{Mbiadou Saleu}, R.~G.}, \bibinfo{author}{Deroussi, L.},
  \bibinfo{author}{Feillet, D.}, \bibinfo{author}{Grangeon, N.}, \&
  \bibinfo{author}{Quilliot, A.} (\bibinfo{year}{2021}).
\newblock \bibinfo{title}{The parallel drone scheduling problem with multiple
  drones and vehicles}.
\newblock {\it \bibinfo{journal}{European Journal of Operational Research}\/},
  {\it \bibinfo{volume}{\text{forthcoming}}\/}.
%Type = Article
\bibitem[{Murray \& Chu(2015)}]{murray2015flying}
\bibinfo{author}{Murray, C.~C.}, \& \bibinfo{author}{Chu, A.~G.}
  (\bibinfo{year}{2015}).
\newblock \bibinfo{title}{The flying sidekick traveling salesman problem:
  Optimization of drone-assisted parcel delivery}.
\newblock {\it \bibinfo{journal}{Transportation Research Part C: Emerging
  Technologies}\/},  {\it \bibinfo{volume}{54}\/}, \bibinfo{pages}{86--109}.
%Type = Article
\bibitem[{Nguyen et~al.(2021)Nguyen, Pham, Hà \& Pham}]{Nguyen2021}
\bibinfo{author}{Nguyen, M.~A.}, \bibinfo{author}{Pham, T.-H.~G.},
  \bibinfo{author}{Hà, M.~H.}, \& \bibinfo{author}{Pham, M.-T.}
  (\bibinfo{year}{2021}).
\newblock \bibinfo{title}{The min-cost parallel drone scheduling vehicle
  routing problem}.
\newblock {\it \bibinfo{journal}{European Journal of Operational Research}\/},
  {\it \bibinfo{volume}{In Press}\/}.
%Type = Article
\bibitem[{Oswin et~al.(2017)Oswin, Fischer, Fischer, Meier, Pferschy, Pilz \&
  Staněk}]{Oswin}
\bibinfo{author}{Oswin, A.}, \bibinfo{author}{Fischer, A.},
  \bibinfo{author}{Fischer, F.}, \bibinfo{author}{Meier, J.~F.},
  \bibinfo{author}{Pferschy, U.}, \bibinfo{author}{Pilz, A.}, \&
  \bibinfo{author}{Staněk, R.} (\bibinfo{year}{2017}).
\newblock \bibinfo{title}{Minimization and maximization versions of the
  quadratic travelling salesman problem}.
\newblock {\it \bibinfo{journal}{Optimization}\/},  {\it
  \bibinfo{volume}{66}\/}, \bibinfo{pages}{521--546}.
%Type = Article
\bibitem[{Padberg \& Rinaldi(1990)}]{Padberg}
\bibinfo{author}{Padberg, M.}, \& \bibinfo{author}{Rinaldi, G.}
  (\bibinfo{year}{1990}).
\newblock \bibinfo{title}{Facet identification for the symmetric traveling
  salesman polytope}.
\newblock {\it \bibinfo{journal}{Math. Program.}\/},  {\it
  \bibinfo{volume}{47}\/}, \bibinfo{pages}{219--257}.
%Type = Unpublished
\bibitem[{Raj et~al.(2021)Raj, Lee, Lee, Walteros \& Murray}]{Raj2021}
\bibinfo{author}{Raj, R.}, \bibinfo{author}{Lee, D.}, \bibinfo{author}{Lee,
  S.}, \bibinfo{author}{Walteros, J.}, \& \bibinfo{author}{Murray, C.}
  (\bibinfo{year}{2021}).
\newblock \bibinfo{title}{A branch-and-price approach for the parallel drone
  scheduling vehicle routing problem}.
\newblock \URLprefix \url{https :// ssrn . com / abstract =3879710}.

\end{thebibliography}
\bibliographystyle{model5-names}

% \appendix
% \section*{Appendix}

\end{document}